\documentclass[11pt,twoside]{article}

\usepackage{parskip}
\usepackage{fullpage}
\usepackage{comment}
\usepackage{epsf}
\usepackage{fancyheadings}
\usepackage{graphics}
\usepackage{graphicx}
\usepackage{psfrag}
\usepackage{color}
\usepackage{amsthm}
\usepackage{amsfonts}
\usepackage{amsmath}
\usepackage{amssymb}
\usepackage{graphicx,epsfig,graphics,epsf}
\usepackage{float,amsmath,amsfonts,amssymb,amsthm}
\usepackage{mathtools}
\usepackage{subcaption}
\usepackage[noend]{algorithmic}
\usepackage[ruled,vlined]{algorithm2e}
\usepackage{url}
\usepackage{fullpage}
\usepackage{makeidx}
\usepackage{enumerate}
\usepackage[top=0.9in, bottom=0.9in, left=0.9in, right=0.9in]{geometry}
\usepackage{graphicx,float,psfrag,epsfig,caption}
\usepackage[usenames,dvipsnames,svgnames,table]{xcolor}

\usepackage{sansmathfonts}
\theoremstyle{plain}

\long\def\comment#1{}

\newcommand{\Prob}{\ensuremath{{\mathbb{P}}}}

\newtheorem{claim}{Claim}[section]
\newtheorem{lemma}[claim]{Lemma}
\newtheorem{fact}[claim]{Fact}

\newtheorem{theorem}{Theorem}

\newtheorem{definition}[claim]{Definition}

 \newenvironment{proofof}[1]{{\bf {\em Proof of #1.}}}{\hfill \qed 
 }

\usepackage{tikz}
\usetikzlibrary{shapes, arrows, calc, positioning,matrix,math}
\tikzset{
data/.style={circle, draw, text centered, minimum height=3em ,minimum width = .5em, inner sep = 2pt},
empty/.style={circle, text centered, minimum height=3em ,minimum width = .5em, inner sep = 2pt},
}
\usepackage{pgfplots}
\pgfplotsset{compat=1.5}
\usepgfplotslibrary{fillbetween}
\usetikzlibrary{patterns}

\newcommand{\G}{{\sf G}}
\newcommand{\Clique}{{\sf Clique}}

\newcommand{\One}{\mathbbm{1}}

\newcommand{\EE}{\ensuremath{{\mathbb{E}}}}
\newcommand{\PP}{\ensuremath{{\mathbb{P}}}}
\newcommand{\QQ}{\ensuremath{{\mathbb{Q}}}}
\DeclareMathOperator{\Var}{Var}

\newlength{\widebarargwidth}
\newlength{\widebarargheight}
\newlength{\widebarargdepth}

\newcommand{\1}{\ensuremath{{\sf (i)}}}
\newcommand{\2}{\ensuremath{{\sf (ii)}}}

\usepackage{bbm}
\usepackage[inline]{enumitem}
\usepackage{enumitem}

\usepackage{tabularx,ragged2e}

\usepackage[pagebackref,letterpaper=true,colorlinks=true,pdfpagemode=none,citecolor=OliveGreen,linkcolor=BrickRed,urlcolor=BrickRed,pdfstartview=FitH]{hyperref}

\begin{document}
\begin{center}
	
	{{\LARGE{ \mbox{Low-degree phase transitions for detecting a planted clique}\\ \mbox{in sublinear time}}}}
	
	\vspace*{.2in}
	
	{\large{
			\begin{tabular}{ccc}
				Jay Mardia$^{\dagger}$, Kabir Aladin Verchand$^{\ddagger,\diamond}$, and Alexander S. Wein$^{\star}$
			\end{tabular}
	}}
	\vspace*{.2in}
	
	\begin{tabular}{c}
		Department of Electrical Engineering$^{\dagger}$, Stanford University\\
		Statistical Laboratory$^{\ddagger}$, University of Cambridge\\
		Schools of Industrial and Systems Engineering$^\diamond$,
		Georgia Institute of Technology\\
		Department of Mathematics$^\star$, University of California, Davis
	\end{tabular}
	
	\vspace*{.2in}

	
	\vspace*{.2in}
	
	\begin{abstract}
		We consider the problem of detecting a planted clique of size $k$ in a random graph on $n$ vertices.  When the size of the clique exceeds $\Theta(\sqrt{n})$, polynomial-time algorithms for detection proliferate.  We study faster---namely, sublinear time---algorithms in the high-signal regime when $k = \Theta(n^{1/2 + \delta})$, for some $\delta > 0$.  To this end, we consider algorithms that non-adaptively query a subset $M$ of entries of the adjacency matrix and then compute a low-degree polynomial function of the revealed entries.  We prove a computational phase transition for this class of \emph{non-adaptive low-degree algorithms}: under the scaling $\lvert M \rvert = \Theta(n^{\gamma})$, the clique can be detected when $\gamma > 3(1/2 - \delta)$ but not when $\gamma < 3(1/2 - \delta)$. As a result, the best known runtime for detecting a planted clique, $\widetilde{O}(n^{3(1/2-\delta)})$, cannot be improved without looking beyond the non-adaptive low-degree class.

		Our proof of the lower bound---based on bounding the conditional low-degree likelihood ratio---reveals further structure in non-adaptive detection of a planted clique.  Using (a bound on) the conditional low-degree likelihood ratio as a potential function, we show that for \emph{every} non-adaptive query pattern, there is a highly structured query pattern of the same size that is at least as effective.
	\end{abstract}
\end{center}

\section{Introduction} \label{sec:introduction}

Many high-dimensional statistical inference problems (e.g., community detection~\cite{decelle2011asymptotic}, planted clique~\cite{jerrum1992large}, and tensor PCA~\cite{richard2014statistical}, to name a few) appear to exhibit statistical-computational gaps wherein the amount (or quality) of data required for all known polynomial-time algorithms may be significantly larger than the amount of data required information-theoretically.  

Central among these is the planted clique problem which we consider here.  In more detail, the planted clique problem consists of observing a graph $G$ on $n$ vertices which may have arisen from one of two distributions: the null distribution in which $G \sim G(n, 1/2)$ (the Erd\H{o}s--R{\'e}nyi distribution) and a planted distribution in which $k$ of the $n$ vertices form a clique and the remaining edges in the graph appear with probability $1/2$, independently. The goal of the detection task is to distinguish these two cases. Information-theoretically, it is possible to detect the presence of a planted clique of size $k \geq (2 + \epsilon) \log_2 (n)$ for any $\epsilon > 0$~\cite{bollobas1976}, whereas the best known polynomial-time algorithms require $k = \Omega(\sqrt{n})$ (see, e.g.,~\cite{kuvcera1995expected,alon1998finding,deshpande2015finding,barak2019nearly}, and the references therein). Each of the aforementioned algorithms requires the full observation of the random graph $G$, which has size $\Theta(n^2)$, and as a result they require runtime at least $\Omega(n^2)$.

When $n$ gets large, even this polynomial running time may prove prohibitively expensive, and it becomes of interest to apply algorithms which run in time sublinear in the input size (see, e.g., the review~\cite{rubinfeld2011sublinear} and references therein). In this work we aim to investigate precisely what runtime is required to detect a clique in the ``easy'' regime $k = \Theta(n^{1/2+\delta})$ for a constant $\delta \in (0,1/2)$. In this regime, the clique vertices can be identified simply based on their degree in the graph~\cite{kuvcera1995expected}, and therefore the maximum degree suffices as a statistic for distinguishing the null and planted distributions. While a naive computation of the maximum degree requires time $\Omega(n^2)$, a faster detection algorithm of runtime $\widetilde{O}(n^{3(1/2-\delta)})$ was given by~\cite{mardia2020finding}: the idea is to approximately estimate the degrees of some subset of the vertices while only examining $\widetilde{O}(n^{3(1/2-\delta)})$ entries of the adjacency matrix. Is this optimal, or might it be possible to reduce the runtime even further?

To explore the fundamental limits of sublinear-time computation, we will consider the \emph{query complexity} of algorithms, that is, the number of entries of the adjacency matrix that need to be read. This, after all, is the bottleneck in the runtime of~\cite{mardia2020finding}. Certainly any algorithm of runtime $O(t)$, for some $t = t(n)$, must make at most $O(t)$ queries. These queries can potentially be chosen adaptively, based on the results of previous queries. On the other hand, the algorithm of~\cite{mardia2020finding} is \emph{non-adaptive}, meaning it specifies upfront a \emph{mask} $M \subseteq {[n] \choose 2}$, i.e., a subset of entries of the input to be observed (depending only on the problem size $n$).

This suggests a natural path forward: prove lower bounds on the query complexity, which in turn imply lower bounds on runtime. In fact this has been studied already: for $k = \Theta(n^{1/2+\delta})$ with $\delta \in (-1/2,1/2)$, it is possible to detect a clique with $\widetilde{O}(n^{2(1/2-\delta)})$ non-adaptive queries, and up to log factors this number of queries is information-theoretically necessary (even if adaptivity is allowed)~\cite{racz2019finding}. This improves the query complexity of~\cite{mardia2020finding}, yet does not lead to a better runtime because a quasipolynomial-time exhaustive search is used to identify a large clique within the queried subgraph. The situation thus proves more subtle than it first appeared: query complexity is not the only bottleneck for runtime.

Our goal will be to show that the runtime of~\cite{mardia2020finding} is optimal, at least within some broad class of algorithms. In light of the above, we cannot merely study the (information-theoretic) query complexity, but will need to further ``tie the hands'' of the algorithm. First, for simplicity we will focus on algorithms that non-adaptively query the input. Second, we will ask that the results of the queries are processed via an efficient (say, polynomial-time) computation. Given the current state of average-case complexity theory, we cannot hope to prove negative results for arbitrary poly-time computation, so we follow a line of prior work~\cite{HS-bayesian,sos-power,hopkins2018statistical} and adopt a popular proxy for this: algorithms that can be represented as $O(\log n)$-degree polynomials. Thus we study the class of \emph{non-adaptive low-degree algorithms}: such an algorithm consists of a sequence (indexed by the problem size $n$) of \emph{masks} $M \subseteq {[n] \choose 2}$ along with a sequence of multivariate polynomials $f: \{0,1\}^M \to \mathbb{R}$ of degree $O(\log n)$ whose input variables are the revealed entries of the adjacency matrix. An algorithm of this type is considered successful at detecting the planted clique if the output of $f$ \emph{separates} (in a sense made precise by Definition~\ref{def: separ}) the null and planted distributions. Logarithmic-degree polynomials are fairly expressive, allowing computation of edge counts, triangle counts, and other small subgraph counts, as well as approximate eigenvalue computations via power iteration (see e.g.,~\cite{kunisky2022notes}). Our lower bound rules out polynomials of even larger degree, namely any $o(\log^2{n})$.

Our main result is to characterize the number of queries $|M| = \Theta(n^\gamma)$ required for a non-adaptive low-degree algorithm to detect a planted clique of size $k = \Theta(n^{1/2+\delta})$ for a constant $\delta \in (0,1/2)$. In more detail, we show (see Theorem~\ref{thm:main} to follow) that---by simulating the degree-counting algorithm of~\cite{mardia2020finding}---some non-adaptive low-degree algorithm succeeds when $\gamma > 3(1/2-\delta)$, but conversely, no non-adaptive low-degree algorithm succeeds when $\gamma < 3(1/2-\delta)$. This lets us complete the phase diagram for planted clique detection in the non-adaptive query model, shown in Figure~\ref{fig:phase-diagram}. As a result, the runtime $\widetilde{O}(n^{3(1/2-\delta)})$ of~\cite{mardia2020finding} for detecting a planted clique cannot be significantly improved without looking beyond the non-adaptive low-degree class.

 \begin{figure}[h!]
	\centering
	\begin{tikzpicture}
\draw[thick] (0, 0) rectangle (8, 4);
\filldraw[fill=BrickRed!20] (0,0) -- (0,4) -- (8,0) -- cycle;
\filldraw[fill=OliveGreen!25] (0,4) -- (6, 2) -- (8,2) -- (8,4) -- cycle;
\filldraw[fill=BurntOrange!20] (0,4) -- (6, 2) -- (8, 2) -- (8,0) -- cycle;
\draw[thick, dashed] (6, 0) -- (6, 2);
\draw[thick] (6, 0) -- (6, -0.1);
\node[below] at (6, 0) {$\gamma = \frac{3}{2}$};
\draw[thick] (8, 0) -- (8, -0.1);
\node[below] at (8, 0) {$\gamma = 2$};
\draw[thick, dashed] (4, 0) -- (4, 2);
\draw[thick] (4, 0) -- (4, -0.1);
\node[below] at (4, 0) {$\gamma = 1$};
\draw[thick] (0, 0) -- (0, -0.1);
\node[below right] at (0, 0) {$\gamma = 0$};
\draw[thick, dashed] (0, 2) -- (6.5, 2); 
\draw[thick] (0, 0) -- (-0.1, 0);
\node[above left] at (0, 0) {$\delta = -\frac{1}{2}$};
\draw[thick] (0, 2) -- (-0.1, 2);
\node[left] at (0, 2) {$\delta = 0$};
\draw[thick] (0, 4) -- (-0.1, 4);
\node[left] at (0, 4) {$\delta = \frac{1}{2}$};

\filldraw[thick, fill=OliveGreen!25] (8.4,3) rectangle (8.7, 3.3); 
\node[right] at (8.8, 3.1) {Easy};

\filldraw[thick, fill=BurntOrange!20] (8.4,2.2) rectangle (8.7, 2.5); 
\node[right] at (8.8, 2.3) {Low-degree hard};

\filldraw[thick, fill=BrickRed!20] (8.4,1.4) rectangle (8.7, 1.7); 
\node[right] at (8.8, 1.5) {Impossible};
\end{tikzpicture}
	\caption{Phase diagram for detecting a clique of size $k = \Theta(n^{1/2+\delta})$ using $|M| = \Theta(n^\gamma)$ non-adaptive queries to the adjacency matrix. If arbitrary computation is allowed on the query results, detection is impossible in the red region and possible otherwise~\cite{racz2019finding}. If a low-degree test must be applied to the query results, our upper bound Theorem~\ref{thm:main}(b) achieves detection in the green (easy) region; in this region, there is also an algorithm for clique detection whose runtime is dominated by the query complexity $\Theta(n^\gamma)$~\cite{mardia2020finding}. Below the line $\delta = 0$, it is known that low-degree polynomials cannot detect the clique, even if the entire input is revealed~\cite{barak2019nearly,hopkins2018statistical}. Our lower bound Theorem~\ref{thm:main}(a) fills in the rest of the hard (yellow) region.}
	\label{fig:phase-diagram}
\end{figure}
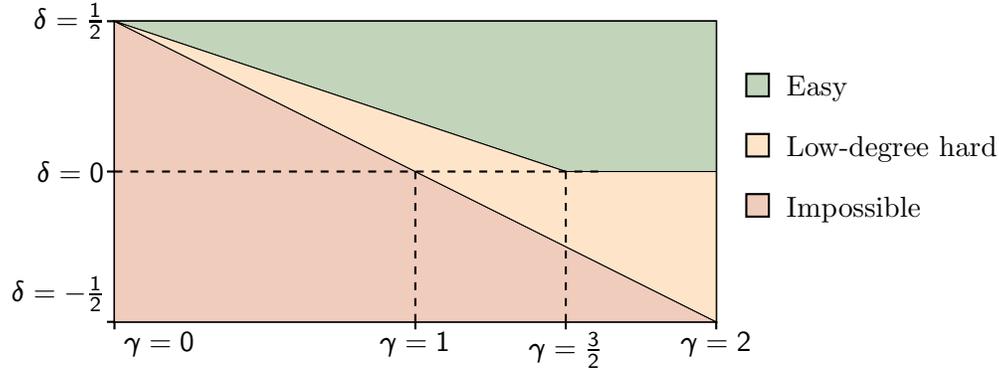

\subsection{Further related work} \label{sec:related}

\paragraph{Computational complexity of statistics.}

Statistical-computational gaps are ubiquitous throughout high-dimensional testing and inference problems. These gaps call for a theory of computational lower bounds (hardness results), as otherwise we can never be sure whether the ``possible but hard'' regime fundamentally admits no efficient algorithm or whether there is a better algorithm waiting to be discovered. For \emph{average-case} computational tasks---where the input is random---we unfortunately lack tools to prove complexity results conditional on standard assumptions such as $P \ne NP$. It is therefore common to resort to one of two tactics: (i) average-case reductions which establish hardness conditional on the hardness of some ``standard'' problem such as planted clique (e.g.,~\cite{brennan18reducibility}) or (ii) proving unconditional failure of particular families of algorithms. Within the latter viewpoint, some popular classes of algorithms to rule out include statistical query (SQ) algorithms (e.g.,~\cite{sq-clique}), the sum-of-squares (SoS) hierarchy (e.g.,~\cite{barak2019nearly}), and low-degree polynomials (the subject of this work).

We discuss briefly the prospect of applying some of the other frameworks mentioned above to our problem of interest---planted clique detection with non-adaptive queries. Average-case reductions (starting from the basic planted clique problem) have in fact already been applied to our setting, showing hardness in the same regime as us but only for certain highly structured masks~\cite{mardia2020finding}; addressing arbitrary masks appears to be beyond the reach of current techniques and we consider this an interesting question for future work. The SQ framework is not directly applicable to our setting because our input (a random graph) does not consist of i.i.d.\ samples; it may be possible to formulate a bipartite variant of our problem in the SQ model, similar to~\cite{sq-clique}. SoS lower bounds tend to be rather unwieldy to prove, even for the basic planted clique problem~\cite{barak2019nearly}, and they only show hardness of the \emph{refutation} problem (which in general need not imply hardness of detection; see~\cite{local-stats,spectral-planting}).

\paragraph{Low-degree polynomials as a model of computation.}

The idea to consider low-degree polynomials as a restricted class of statistical tests first arose from the sum-of-squares literature~\cite{barak2019nearly,HS-bayesian,sos-power,hopkins2018statistical} and has by now found success in a wide variety of settings (see~\cite{kunisky2022notes} for a survey), including extensions beyond hypothesis testing~\cite{ld-opt,schramm2022computational}. For instance, in the planted clique problem, $O(\log n)$-degree polynomials can detect a clique of size $k \gtrsim \sqrt{n}$ but provably fail to detect a clique of size $k \ll \sqrt{n}$~\cite{barak2019nearly,hopkins2018statistical}, suggesting that this threshold is a fundamental barrier for efficient computation (or more conservatively, a barrier for certain known approaches). For planted clique and various other inference problems of this style, low-degree polynomials capture the best known poly-time algorithms and give a rigorous explanation for apparent computational barriers.

Our work is the first to employ low-degree polynomials to probe the precise limits of sublinear computation. Prior work has addressed coarser questions about runtime by taking polynomial degree as a proxy for runtime, e.g., with degree $n^\delta$ corresponding to time $\exp(n^{\delta \pm o(1)})$ (see e.g.,~\cite{ding2023subexponential}). In our regime of sublinear runtime, we cannot hope for a meaningful correspondence between polynomial degree and runtime, since even a degree-1 polynomial can already read the entire input. Instead, our approach relies on explicitly restricting the algorithm to a small fraction of the input variables.

On a technical level, we use the standard \emph{low-degree likelihood ratio} (see~\cite{hopkins2018statistical}) as a tool for ruling out all low-degree polynomial tests. For testing between a specific pair of planted and null distributions, this often boils down to a relatively straightforward computation. However, since we allow an arbitrary choice of mask, we effectively need to prove \emph{many} such hardness results all at once. To complicate things further, our setting requires a \emph{conditional} variant of the low-degree likelihood ratio~\cite{fp,coja2022statistical,dhawan2023detection,ding2023low}. We give an overview of the proof in Section~\ref{sec:proof-overview}.

\paragraph{Average-case fine-grained hardness.} In analogy with classical worst-case to average-case reductions (see e.g.,~\cite{ajtai1996generating,regev2010learning}), a recent line of literature establishes fine-grained notions of average-case complexity~\cite{ball2017average,dalirrooyfard2020new} via worst-case to average-case reductions for problems such as counting cliques in random hypergraphs~\cite{goldreich2018counting,boix2021average} and counting bicliques in random bipartite graphs~\cite{hirahara2021nearly}.  Our work departs from these along two axes: First, our hardness results are unconditional but restricted to algorithms which can be expressed as low-degree polynomials. Second, as opposed to a counting problem, here we consider a testing problem which appears to exhibit a statistical-computational gap. To the best of our knowledge, our work (along with our previous work~\cite{mardia2020finding}) is the first to provide such evidence for fine-grained hardness of testing.

\paragraph{Restricting algorithms via query complexity.} While restricting algorithms via query complexity forms a dominant theme in the study of sublinear-time algorithms and property testing (see, e.g.,~\cite{rubinfeld2011sublinear} and~\cite{goldreich2017intro} for comprehensive accounts), we note that the restriction to algorithms which make non-adaptive queries can be alternatively motivated in its own right (without reference to sublinear runtime).  In particular, this models a scenario where the statistician must decide upfront which data to collect. This is relevant in (for instance) the \emph{group testing} problem where the goal is to identify which individuals are afflicted by a disease based on ``pooled'' tests (see, e.g.,~\cite{aldridge2019group} for a survey). It is realistic to assume that the subset of individuals included in each test must be chosen non-adaptively (without knowledge of other test results) for purposes of practical implementation. Unlike our problem, there is no statistical-computational gap in group testing: a particular choice for the ``design'' (the choice of subsets to test, analogous to our ``mask'') succeeds using the information-theoretic minimum number of tests~\cite{coja2020optimal}.  

Turning to query complexity in problems on random graphs,~\cite{ferber2016finding,conlon2020online,alweiss2021subgraph} consider the \emph{subgraph query problem}, that is, the query complexity of finding a fixed subgraph in a sufficiently large Erd\H{o}s--R{\'e}nyi random graph.  Following this line of work,~\cite{feige2020finding} study the query complexity of finding a large clique in an Erd\H{o}s--R{\'e}nyi graph.  Importantly,~\cite{feige2020finding} allow for a limited amount of adaptivity in their queries, allowing a constant number of rounds in which the queries in each round may depend on the result of previous rounds.  Tighter bounds on the query complexity upon restricting the number of rounds of adaptivity were later obtained by~\cite{feige2021tight,csoka2023finding}.  

Information-theoretic (potentially adaptive) query complexity limits were established for the planted variant of the problem which we study here by~\cite{racz2019finding} (see also~\cite{rashtchian2021average} for an alternate proof of the lower bound via communication complexity).  The more general problem of random subgraph detection was later studied by~\cite{huleihel2021random} in which an analogous information-theoretic threshold of $\widetilde{\Theta}(n^2/k^2)$ was established as well as a polynomial-time algorithm which requires $\widetilde{\Omega}(n^3/k^3)$ queries to succeed. In contrast with this last line of work, we consider the planted variant of the problem and provide \emph{restricted} query complexity lower bounds over the family of algorithms which can be expressed as low-degree polynomial functions of the input.  These lower bounds in turn match existing algorithmic upper bounds in the literature~\cite{mardia2020finding}.

\subsection{Open problems} \label{sec:open}
A number of interesting directions remain open and we detail a few here.  
\begin{itemize}
    \item \emph{Reaching the computational threshold.}  While our results indicate a smooth tradeoff between clique size and runtime above the computational threshold, it is less clear what happens near the computational threshold $k \approx \sqrt{n}$. When $k = \omega(\sqrt{n \log n})$, the algorithm of~\cite{mardia2020finding} runs in time $\widetilde{O}(n^{3/2})$.  However, when the clique has size $k = o(\sqrt{n \log n})$ so that degree counting no longer works, but is still above the computational threshold $k = \Omega(\sqrt{n})$, it is unknown whether detection is possible in \emph{strongly sublinear} time, i.e., time $O(n^{2-\epsilon})$ for a constant $\epsilon > 0$.
  
    \item \emph{Adaptivity.} Our results provide a lower bound on the capability of non-adaptive low-degree sublinear-time algorithms to detect a planted clique. It is natural to wonder whether adaptivity helps, or whether our lower bounds can be extended to adaptive algorithms. One fruitful direction may be to define a restricted family of adaptive algorithms---perhaps one where low-degree polynomials govern how queries are selected adaptively---and prove lower bounds against this family. We note that while adaptivity does not appear to help for detection, it does appear to help for the related problem of recovering the clique vertices. That is, the degree counting algorithms of~\cite{mardia2020finding} for detection are non-adaptive, but their natural counterparts for recovery require adaptive queries.

    \item \textit{Beyond low-degree polynomials.} Our result provides evidence for a query-complexity based statistical-computational gap for the planted clique problem. It would be nice to gain more evidence for this by proving similar hardness results for algorithmic classes beyond low-degree polynomials. We believe the series of steps involved in our impossibility result should be useful even when considering other algorithmic classes. At the end of Section~\ref{sec:proof-overview} we briefly mention which aspect of our approach is specific to low-degree polynomials. It would be very interesting to implement this step for other algorithmic classes. Our ideal goal would be to show that the algorithm of \cite{mardia2020finding} is essentially optimal, conditional on the planted clique conjecture.
    
\end{itemize}

\section{Problem setup and main results} \label{sec:main-res}

As previously alluded to, we will consider \emph{non-adaptive low-degree algorithms}. When the full input is a graph on $n$ vertices, a non-adaptive algorithm must specify the subgraph to be queried before any observations are made.  We formalize this through the notion of a mask, defined presently.

\begin{definition}[Mask and mask degree]\label{def:mask}\ \\
	A mask $M$ over a ground set $[n]$ is a subset of ${[n] \choose 2}$. We interpret this subset as specifying a graph on vertex set $V(M) := \left\{v \in [n] : \exists u \in [n] \text{ such that } (u,v) \in M \right\}$ with edge set $M$. Whenever the ground set is clear from context, we will not mention it and suppress it in our notation.  For any vertex $v \in V(M)$ we refer to its mask degree as ${\sf deg}^M(v) := \left\lvert \left\{ u \in V(M) : (u,v) \in M \right\} \right\rvert$.
\end{definition}
Equipped with this definition, we turn to our null distributions $G(n, 1/2)$ and $G(n, M)$, which denote the Erd\H{o}s--R{\'e}nyi distribution and its masked counterpart, respectively.

\begin{definition} [Erd\H{o}s--R{\'e}nyi distribution $G(n, 1/2)$ and masked Erd\H{o}s--R{\'e}nyi distribution $G(n, M)$] \label{def:er}
	$G(n, 1/2)$ is the uniform distribution on $\{+1, -1\}^{\binom{n}{2}}$.  We interpret a sample from this distribution as describing a graph $G$ with vertex set $[n]$ in which the edge $(i, j)$ is present if and only if the entry indexed by the unordered pair $(i, j)$ is $+1$. Moreover, given a mask $M$, the masked Erd\H{o}s--R{\'e}nyi distribution $G(n, M)$ denotes the marginal distribution of $G(n, 1/2)$ restricted to the coordinates in $M$. 
\end{definition}

Turning to our planted distributions, we first define the clique distribution and its conditional counterpart (which will be important as a proof device).
\begin{definition}[Clique indicator distributions: $\Clique(n, k)$ and $\Clique(n, k, S)$]\label{def:cl-ind}\ \\
	$\Clique(n, k)$ denotes the uniform distribution over vectors in $\{0,1\}^{[n]}$ that have exactly $k$ nonzero coordinates. We interpret a sample from this distribution as a choice of which vertices belong to the planted clique. For any subset $S \subseteq [n]$, $\Clique(n, k, S)$ denotes the distribution $\Clique(n, k)$ conditioned on all nonzero coordinates being inside $S$.
\end{definition}

Now we define our planted distributions $G(n, 1/2, k)$ and $G(n, k, M)$.

\begin{definition}[Planted Clique distributions: $G(n, 1/2, k)$ and $\G(n,k,M)$]\label{def:mask-pc}\ \\
	Let $K$ be sampled from $\Clique(n, k)$. Then $\G(n, 1/2,k)$ is the following distribution on $\{+1,-1\}^{[n] \choose 2}$.
	\begin{enumerate}
		\item Coordinates corresponding to unordered pairs $(i,j)$ with $K_i = K_j = 1$ are $+1$.
		\item Every other coordinate is independent and uniform on $\{+1,-1\}$.
	\end{enumerate}
    We again interpret a sample from this distribution as a graph on vertex set $[n]$, where $+1$ denotes the presence of an edge. When restricted to the coordinates in the mask $M$, the distribution is denoted by $G(n, k, M)$.
\end{definition}
Given a graph $G$ restricted to the coordinates in the mask $M$, our task, then, is to determine whether the observations originated from the null distribution $G(n, M)$ or the planted distribution $G(n, k, M)$.

Our main result provides a tight condition on the size of the mask $M$ which controls whether or not \emph{separation}---defined presently---between the null distribution $G(n, M)$ and the planted distribution $G(n, k, M)$ is possible using a low-degree test. Our results are asymptotic in nature, and thus we will often refer to a \emph{sequence} of problems (or algorithms, etc.), which are assumed to be indexed by the problem size $n$.

\begin{definition}[Strong/weak separation]\label{def: separ}
	Consider a sequence $N = N_n$ and two (sequences of) distributions $\mathbb{P}$ and $\mathbb{Q}$ on $\mathbb{R}^N$. A sequence of polynomials $f: \mathbb{R}^N \rightarrow \mathbb{R}$ separates $\mathbb{P}$ and $\mathbb{Q}$ weakly if as $n \rightarrow \infty$,
	\begin{align}\label{def:weak-separation}
		\Bigl(\max\bigl\{\Var_{\PP}(f), \Var_{\QQ}(f)\bigr\}\Bigr)^{1/2} = O\bigl( \bigl \lvert \EE_{\PP}[f] - \EE_{\QQ}[f] \bigr \rvert \bigr), \tag{Weak separation}
	\end{align}
	and strongly if as $n \rightarrow \infty$, 
	\begin{align}\label{def:strong-separation}
		\Bigl(\max\bigl\{\Var_{\PP}(f), \Var_{\QQ}(f)\bigr\}\Bigr)^{1/2} = o\bigl( \bigl \lvert \EE_{\PP}[f] - \EE_{\QQ}[f] \bigr \rvert \bigr). \tag{Strong separation}
	\end{align}
\end{definition}

As is standard in the low-degree testing literature, we take separation as the definition of ``success'' for low-degree tests. Separation is a natural sufficient condition that allows two distributions to be distinguished using the output of a polynomial. Specifically, strong separation implies (by Chebyshev's inequality) that $\PP,\QQ$ can be distinguished with probability $1-o(1)$, and weak separation implies that $\PP,\QQ$ can be distinguished with nontrivial advantage over a random guess; see~\cite{fp}. We are now in position to state our main result.  

\begin{theorem} \label{thm:main}
	Fix constants $0 < \delta < 1/2$ and $0 < \gamma < 2$. Consider a sequence $k =  \Theta(n^{1/2+ \delta})$.
	\begin{itemize}
		\item[(a)] (Lower bound) If $\gamma < 3(1/2 - \delta)$ then for any sequence of masks with $|M| = O(n^\gamma)$, any sequence of degree-$o(\log^2 n)$ polynomials fails to weakly separate $G(n,M)$ and $G(n,k,M)$.
		\item[(b)] (Upper bound) If $\gamma > 3(1/2 - \delta)$ then there exists a sequence of masks with $|M| = O(n^\gamma)$ and a sequence of polynomials with constant degree $C(\gamma, \delta)$ that strongly separates $G(n,M)$ and $G(n,k,M)$.
	\end{itemize}
\end{theorem}

We provide the proof of the lower bound in Section~\ref{sec:proof-lb} and the proof of the upper bound in Section~\ref{sec:proof-ub}. The upper bound essentially simulates the degree counting algorithm of~\cite{mardia2020finding} using polynomials. The lower bound is our main contribution, and we now turn to an overview of its proof.

\subsection{Overview of the lower bound proof}\label{sec:proof-overview}

The main challenge in proving our hardness result lies in establishing the failure of low-degree polynomials for an \textit{arbitrary} mask with a small number of (at most $O(n^\gamma)$) edges.

\subsection*{Step 1: Reducing to masks with small maximum mask degree}

We gain intuition about algorithmically useful masks by studying~\cite{mardia2020finding}'s sublinear-time algorithm.  This algorithm is based on~\cite{kuvcera1995expected}'s observation that for large planted clique sizes (e.g.\ $k = n^{1/2 + \delta}$), with high probability the degree of all planted clique vertices is much larger than the degree of all  non-clique vertices.  Consider the first $(n/k)\cdot \mathsf{polylog}(n)$ vertices. With high probability, if the graph were drawn from the planted distribution $G(n, 1/2, k)$ (see Definition~\ref{def:mask-pc}), then at least one of these vertices will belong to the planted clique and have large degree.  By contrast, if the graph were drawn from the null distribution $G(n, 1/2)$, all of these vertices will have small degree.

Simply estimating the degree of each of these $(n/k) \cdot \mathsf{polylog}(n)$ vertices and checking if any of them is `large enough' to be a planted clique vertex will let us distinguish between the null and planted cases. \cite{mardia2020finding}'s observation was that even just an estimate of these degrees obtained by subsampling $O(n^2/k^2)$ potential neighbours (instead of computing the degree exactly by looking at all $n-1$ potential neighbours) is good enough to distinguish between the planted and null cases.

Clearly any mask that allows for such an estimate for enough vertices will be algorithmically useful. Luckily, since $k = \Theta(n^{1/2+\delta})$ and  $\gamma < 3(1/2-\delta)$, our mask does not have enough mask edges to query $\Omega(n^3/k^3)$ entries of the adjacency matrix. In fact, this means that only $o(n/k)$ vertices can hope to have a `large' mask degree $\Omega(n^2/k^2)$. But with probability $1-o(1)$, because the planted clique vertices are chosen uniformly at random, none of these vertices with large mask degree will be planted vertices even in the planted case.

As a result, for masks with too few ($O(n^\gamma)$) mask edges, we can safely ignore vertices with `large' mask degree, as those vertices will behave the same under both the null and planted distributions (with high probability), and intuitively this means they carry no useful information.  Hence it suffices to show hardness just for masks with a small maximum mask degree. Formally, we implement this reduction via the conditional low-degree likelihood method (see, e.g.,~\cite{fp}, Proposition 6.2): we condition on the high-probability event that no vertices of `large' mask degree belong to the planted clique\footnote{Conditioning is necessary here: Consider the mask which consists solely of all edges connected to the first vertex.  This mask is clearly insufficient to detect the planted clique, yet the corresponding low-degree likelihood ratio blows up.}.

\subsection*{Step 2: Reducing to masks with few mask vertices}
Consider two masks, both of size $\widetilde{\Theta}(n^3/k^3)$, the first mask being any mask of this size that lets us estimate enough degrees as well as needed for \cite{mardia2020finding}'s algorithm, and the second being a `square' mask consisting of all potential mask edges involving only a fixed set of $\widetilde{\Theta}(n^{3/2}/k^{3/2})$ vertices. We know that the former mask is algorithmically useful, but the latter mask is not, as we later discuss in Step 3.

To express our takeaway from this, we should consider the mask degree distribution. For two masks with the same number of mask edges, the average of this distribution will be the same. However, the example above indicates that for masks with the same average mask degree, having a `more skewed' (lots of large mask degrees as well as small mask degrees) degree distribution is more useful than having a `more uniform' one.

Within the context of separation by low-degree polynomials, `algorithmic utility' can be quantified by the norm of the low-degree likelihood ratio\footnote{Technically in all our lemmas we will work with a natural upper bound to this quantity, but that detail is unimportant for this overview.}, which is a standard quantity (see e.g.,~\cite{hopkins2018statistical}) defined in~\eqref{eq:ld-norm-def}. We will use the intuition above to upper bound the `algorithmic utility' (norm of the low-degree likelihood ratio) of a mask $M$ by that of a closely related mask $M'$. We will obtain the mask $M'$ through a small tweak to $M$ that slightly skews its mask degree distribution while keeping the total number of mask edges the same. 

If we repeat this process iteratively and use the fact that \textit{we are only interested in masks with small maximum mask degree}, we can show that the `algorithmic utility' of every such mask is upper bounded by the `algorithmic utility' of a mask with only a few vertices. This reduction is the main technical contribution of our work. 

\subsection*{Step 3: Analytically upper bounding the norm of the low-degree likelihood ratio for masks with few mask vertices}

Once we know we only need to upper bound the `algorithmic utility' (norm of the low-degree likelihood ratio) of masks with few vertices, we are ready to conclude.  In particular, calculating such an upper bound analytically proves tractable using standard techniques, which gives our desired result.

\textbf{Remark:} Of the three steps above, Step 1 and Step 3 have natural analogues even when considering algorithmic classes other than low-degree polynomials.  Only Step 2 seems to crucially rely on properties of low-degree polynomials. 
As alluded to in Section~\ref{sec:open}, it would be interesting to see if the above reduction program can be carried out for other algorithmic classes. In particular, it would be nice to show that---under the planted clique conjecture---the algorithm of \cite{mardia2020finding} is essentially optimal among non-adaptive algorithms.

\section{Proof of the lower bound: Theorem~\ref{thm:main}(a)} \label{sec:proof-lb}

In this section, we provide a sequence of lemmas implementing the strategy outlined in Section~\ref{sec:proof-overview}.  This section culminates in the proof of the lower bound.  We first require the following two definitions.

\begin{definition}[Low-degree likelihood ratio upper bound: ${\sf LDUB}(n,M)$]\label{def: ldub}\ \\
	Given integers $n$, $k \leq n$, $D$ and a mask $M$ on ground set $[n]$, let $X$ and $X'$ be two independent draws from $\Clique(n, k)$ (as in Definition~\ref{def:cl-ind}). For $i \in [n]$, let $Z_i := X_i \cdot X'_i$. The low-degree likelihood ratio upper bound at degree $D$ is defined as
	\[
	{\sf LDUB}(n,M) := 1 + \sum\limits_{d = 1}^{D} \frac{1}{d!}\cdot \mathbb{E}\biggl[\Bigl(\sum\limits_{(i,j) \in M}Z_i \cdot Z_j\Bigr)^d\biggr] .
	\]
	The values of $k$ and $D$ will always be clear from context, so we suppress them and denote the quantity as just ${\sf LDUB}(n,M)$ to simplify notation.
\end{definition}

\begin{definition}[Conditional low-degree likelihood ratio upper bound: ${\sf Cond}(n,M,S)$]\label{def: cldub}\ \\
	Given integers $n$, $k \leq n$, $D$, a mask $M$ on ground set $[n]$, and a subset $S \subseteq [n]$, the conditional low-degree likelihood ratio upper bound ${\sf Cond}(n,M,S)$ is defined analogously to the low-degree likelihood ratio upper bound ${\sf LDUB}(n,M)$ with one crucial difference. The independent random vectors $X$ and $X'$ (as in Definition~\ref{def: ldub}) are drawn from $\Clique(n, k, S)$ (rather than $\Clique(n, k)$). The rest of the definition proceeds as in Definition~\ref{def: ldub}.\footnote{Cleary, ${\sf Cond}(n,M,[n]) = {\sf LDUB}(n,M)$.}
\end{definition}

The motivation behind these definitions is the following. There is a standard quantity, the \emph{norm of the (conditional) degree-$D$ likelihood ratio}, to be defined in~\eqref{eq:ld-norm-def}. It is well known that if this quantity is $1+o(1)$ then this implies our goal: degree-$D$ polynomials cannot achieve weak separation; see~\cite[Proposition~6.2]{fp}. Following~\cite[Proposition~B.1]{spectral-planting}, ${\sf Cond}(n,M,S)$ provides a convenient upper bound on this quantity.

\begin{algorithm}[H]
	\SetAlgoLined
	\tcp{Give as many edges from $v$ to $u$ as possible}
	\KwIn{Mask $M$, donating vertex $v \in V(M)$, receiving vertex $u \in V(M)$}
	\KwOut{Mask ${\sf Donate}(M,v \rightarrow u)$}
	\textbf{initialize} ${\sf Donate}(M,v \rightarrow u) = M$
	
	\For{ vertex $s \in V(M) \setminus \{v,u\}$}{\If{ edge $(s,v) \in M$ and edge $(s,u) \notin M$ }{
			
			remove $(s,v)$ from ${\sf Donate}(M,v \rightarrow u)$
			
			add $(s,u)$ to ${\sf Donate}(M,v \rightarrow u)$ }
		
	}
	\textbf{output} ${\sf Donate}(M,v \rightarrow u)$

	\caption{ ${\sf Donate}(M,v \rightarrow u)$}
	\label{alg:donation}
\end{algorithm}

\begin{fact}[Donation leaves certain mask properties (almost) unchanged]\label{fact: conservation-donation}\ \\
	Let $M$ be a mask with vertices $v,u \in V(M)$. It is an easy observation about Algorithm~\ref{alg:donation} that
	\begin{enumerate}
		\item ${\sf Donate}(M,v \rightarrow u)$ has the same number of edges as $M$.
		\item $V({\sf Donate}(M,v \rightarrow u))$ must be either $V(M)$ or $V(M) \setminus v$.\\
	\end{enumerate}
\end{fact}

\begin{figure}[h!]
	\centering

\begin{tikzpicture}
	
	\node[circle,fill=black,inner sep=0pt,minimum size=5pt] at (-4.5, 1) {};
	\node[color=black, above] at (-4.5, 1.02) {$v$};
	\node[circle,fill=black,inner sep=0pt,minimum size=5pt] at (-2, 1) {};
	\node[color=black, above] at (-2, 1.02) {$u$};
	\node[circle,fill=black,inner sep=0pt,minimum size=5pt] at (-6.5, -1) {};
	\node[circle,fill=black,inner sep=0pt,minimum size=5pt] at (-6.3, -1) {};
	\node[circle,fill=black,inner sep=0pt,minimum size=5pt] at (-6.1, -1) {};
	\node[circle,fill=black,inner sep=0pt,minimum size=5pt] at (-5.9, -1) {};
	\draw[] (-4.5, 1) -- (-6.5, -1);
	\draw[] (-4.5, 1) --(-6.3, -1);
	\draw[] (-4.5, 1) --(-6.1, -1);
	\draw[] (-4.5, 1) --(-5.9, -1);
	
	\node[circle,fill=black,inner sep=0pt,minimum size=5pt] at (-3.55, -1) {};
	\node[circle,fill=black,inner sep=0pt,minimum size=5pt] at (-3.35, -1) {};
	\node[circle,fill=black,inner sep=0pt,minimum size=5pt] at (-3.15, -1) {};
	\node[circle,fill=black,inner sep=0pt,minimum size=5pt] at (-2.95, -1) {};
	\draw[] (-4.5, 1) --(-3.55, -1);
	\draw[] (-4.5, 1) --(-3.35, -1);
	\draw[] (-4.5, 1) --(-3.15, -1);
	\draw[] (-4.5, 1) --(-2.95, -1);
	
	\draw[] (-2, 1) --(-3.55, -1);
	\draw[] (-2, 1) --(-3.35, -1);
	\draw[] (-2, 1) --(-3.15, -1);
	\draw[] (-2, 1) --(-2.95, -1);
	
	\node[circle,fill=black,inner sep=0pt,minimum size=5pt] at (-0.2, -1) {};
	\node[circle,fill=black,inner sep=0pt,minimum size=5pt] at (-0.4, -1) {};
	\node[circle,fill=black,inner sep=0pt,minimum size=5pt] at (-0.6, -1) {};
	\node[circle,fill=black,inner sep=0pt,minimum size=5pt] at (-0.8, -1) {};
	
	\draw[] (-2, 1) --(-0.2, -1);
	\draw[] (-2, 1) --(-0.4, -1);
	\draw[] (-2, 1) --(-0.6, -1);
	\draw[] (-2, 1) --(-0.8, -1);
	
    \draw (-5.6, -0.25) ellipse (0.5cm and 0.25cm);
    \node[color=black, below] at (-6.5, -1.1) {$N_v$};
    \draw [-stealth] (-5.6, -0.5) to [out=-50,in=-130] (2.5, -0.75);
	\node[color=black, above right] at (-0.3, -1) {$N_u$};
	\node[color=black, below] at (-3.25, -1.1) {$\mathsf{Both}$};
	\node[color=black, below] at (-3.25, -2.25) {$\text{(a)}$};
	
	\hfill
	\node[circle,fill=black,inner sep=0pt,minimum size=5pt] at (5.5, 1) {};
	\node[color=black, above] at (5.5, 1.02) {$u$};
	\node[circle,fill=black,inner sep=0pt,minimum size=5pt] at (3, 1) {};
	\node[color=black, above] at (3, 1.02) {$v$};
	\node[circle,fill=black,inner sep=0pt,minimum size=5pt] at (1.8, -1) {};
	\node[circle,fill=black,inner sep=0pt,minimum size=5pt] at (1.6, -1) {};
	\node[circle,fill=black,inner sep=0pt,minimum size=5pt] at (1.4, -1) {};
	\node[circle,fill=black,inner sep=0pt,minimum size=5pt] at (1.2, -1) {};
	\node[circle,fill=black,inner sep=0pt,minimum size=5pt] at (7.1, -1) {};
	\node[circle,fill=black,inner sep=0pt,minimum size=5pt] at (6.9, -1) {};
	\node[circle,fill=black,inner sep=0pt,minimum size=5pt] at (6.7, -1) {};
	\node[circle,fill=black,inner sep=0pt,minimum size=5pt] at (6.5, -1) {};
	\draw[] (5.5, 1) -- (1.4, -1);
	\draw[] (5.5, 1) --(1.2, -1);
	\draw[] (5.5, 1) --(7.1, -1);
	\draw[] (5.5, 1) --(6.9, -1);
	\draw[] (5.5, 1) -- (1.6, -1);
	\draw[] (5.5, 1) -- (1.8, -1);
	\draw[] (5.5, 1) -- (6.7, -1);
	\draw[] (5.5, 1) -- (6.5, -1);
	
	\node[circle,fill=black,inner sep=0pt,minimum size=5pt] at (4.55, -1) {};
	\node[circle,fill=black,inner sep=0pt,minimum size=5pt] at (4.35, -1) {};
	\node[circle,fill=black,inner sep=0pt,minimum size=5pt] at (4.15, -1) {};
	\node[circle,fill=black,inner sep=0pt,minimum size=5pt] at (3.95, -1) {};
	\draw[] (5.5, 1) --(4.55, -1);
	\draw[] (5.5, 1) --(4.35, -1);
	\draw[] (5.5, 1) --(4.15, -1);
	\draw[] (5.5, 1) --(3.95, -1);
	
	\draw[] (3, 1) --(4.55, -1);
	\draw[] (3, 1) --(4.35, -1);
	\draw[] (3, 1) --(4.15, -1);
	\draw[] (3, 1) --(3.95, -1);

	\node[color=black, below] at (1.2, -1.1) {$N_v$};
	\node[color=black, below] at (7.25, -1.1) {$N_u$};
	\node[color=black, below] at (4.25, -1.1) {$\mathsf{Both}$};
	\node[color=black, below] at (4.25, -2.25) {$\text{(b)}$};
	
\end{tikzpicture}
	\caption{Illustration of the $\mathsf{Donate}$ process (Algorithm~\ref{alg:donation}) using notation from the proof of Lemma~\ref{lem: donation}. Panel (a) shows the neighborhoods of $u$ (denoted as $N_u \cup \mathsf{Both}$) and $v$ (denoted as $N_v \cup \mathsf{Both}$) before $\mathsf{Donate}(v \rightarrow u)$.  Panel (b) shows the result: Each neighbor of $v$ not originally connected to $u$ (i.e.\ in $N_v$) is disconnected from $v$ and connected to $u$ instead.}
	\label{fig:swap-mask}
\end{figure}
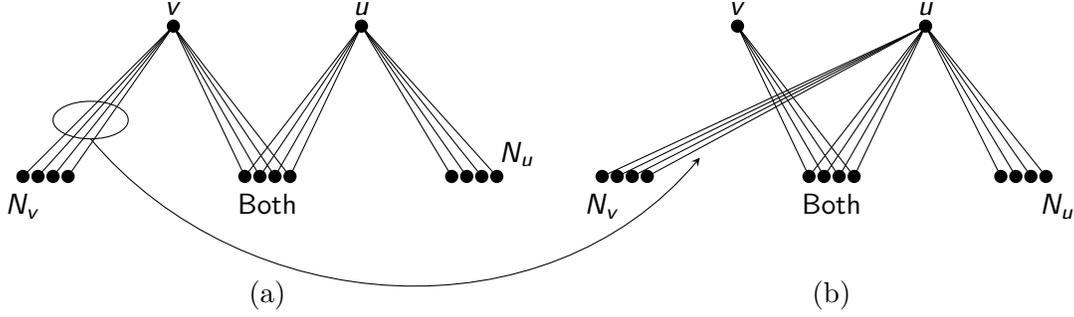

\begin{lemma}[Donation cannot hurt low-degree algorithms]\label{lem: donation}\ \\
	Let $M$ be a mask with vertices $v,u \in V(M)$. Then, informally, donation from $v$ to $u$ (Algorithm~\ref{alg:donation}) cannot decrease the low-degree likelihood ratio upper bound. Formally,
	\[{\sf LDUB}(n,M) \leq  {\sf LDUB}(n,{\sf Donate}(M,v \rightarrow u)).\]
\end{lemma}
\begin{proof}
	Let $M' = {\sf Donate}(M,v \rightarrow u)$. Denote $\Phi(M) := \sum\limits_{(i,j) \in M}Z_i \cdot Z_j$ and $\Phi(M') := \sum\limits_{(i,j) \in M'}Z_i \cdot Z_j$, where the $Z_i$'s are as in Definition~\ref{def: ldub}.
	We will show that for any positive integer $d$, \[\mathbb{E}\left[\left(\Phi(M)\right)^d\right] \leq \mathbb{E}\left[\left(\Phi(M')\right)^d\right],\] since the desired conclusion follows easily from this.
	
	In fact, by the law of total expectation, it suffices to simply show the following inequality for any set $\{z_i \in \{0,1\} : i \in [n] \setminus \{v,u\}\}$ in the suppport of $\{Z_i : i \in [n] \setminus \{v,u\}\}$:
	
	\begin{equation}\label{eqn:exp-conditioned}
		\mathbb{E}\left[\left(\Phi(M)\right)^d \big\vert Z_i = z_i, \forall i \in [n] \setminus \{v,u\} \right] \leq \mathbb{E}\left[\left(\Phi(M')\right)^d \big\vert Z_i = z_i, \forall i \in [n] \setminus \{v,u\} \right].
	\end{equation}
	
	Hence, for the rest of the proof, we condition on $Z_i = z_i$ for $i \in [n] \setminus \{v,u\}$ and the only randomness is in the random variables $Z_v,Z_u$. For the rest of the proof when we take expectations, even though the aforementioned conditioning exists, we will not indicate it notationally.
	\begin{enumerate}
		\item Let $N_u$ be the vertices in $V(M) \setminus \{ v,u\}$ connected (by edges in $M$) to $u$ but not to $v$.
		\item Let $N_v$ be the vertices in $V(M) \setminus \{ v,u\}$ connected (by edges in $M$) to $v$ but not to $u$.
		\item Let ${\sf Both}$ be the vertices in $V(M) \setminus \{ v,u\}$ connected (by edges in $M$) to both $u$ and $v$.
		\item Let ${\sf Other}$ be the set of edges in $M$ involving neither $u$ nor $v$.
	\end{enumerate}

	With this notation and the fact that $\One_{(u,v) \in M'} = \One_{(u,v) \in M}$ , we can rewrite \[\Phi(M) = \sum\limits_{s \in N_v}z_s \cdot Z_v + \sum\limits_{s \in N_u}z_s \cdot Z_u + \sum\limits_{s \in {\sf Both}}z_s \cdot (Z_u+Z_v)+   \sum\limits_{(i,j) \in {\sf Other}}z_i \cdot z_j + \One_{(u,v) \in M} \cdot Z_uZ_v\] and \[\Phi(M') = \sum\limits_{s \in N_u \cup N_v}z_s \cdot Z_u + \sum\limits_{s \in {\sf Both}}z_s \cdot (Z_u+Z_v) +  \sum\limits_{(i,j) \in {\sf Other}}z_i \cdot z_j + \One_{(u,v) \in M} \cdot Z_uZ_v.\]
	
	It is clear from the expressions above that if $Z_u = Z_v$, we have $\Phi(M) = \Phi(M')$, and hence
	\[\mathbb{E}\left[\left(\Phi(M)\right)^d \big\rvert Z_u = Z_v \right] = \mathbb{E}\left[\left(\Phi(M')\right)^d \big\rvert Z_u = Z_v \right].\]
	
	This fact combined with the law of total expectation means we now only need to prove
	\[\mathbb{E}\left[\left(\Phi(M)\right)^d \big\rvert Z_u \neq Z_v \right] \leq \mathbb{E}\left[\left(\Phi(M')\right)^d \big\rvert Z_u \neq Z_v \right],\] assuming the event $Z_u \neq Z_v$ has positive conditional probability (if it does not, we are already done).
	
	For the rest of the proof, assume $Z_u \neq Z_v$. We must have $Z_u+Z_v = 1$ and $Z_uZ_v=0$. As a result, there exist non-negative constants $c$, $c_u$, and $c_v$ (which depend on the values of $z_i$ for $i \in [n] \setminus \{v,u\}$) such that
	\[\Phi(M) = c_u \cdot Z_u + c_v \cdot Z_v + c\] and \[\Phi(M') = (c_u+c_v) \cdot Z_u + c.\]
	
	By symmetry, the events $Z_u = 1, Z_v=0$ and $Z_u=0,Z_v=1$ both have probability $1/2$ (as long as the events we have conditioned on so far occur with positive probability). Hence,
	\[ \mathbb{E}\left[\left(\Phi(M)\right)^d \big\rvert Z_u \neq Z_v \right] = \frac{\left(c_u+c\right)^d+\left(c_v+c\right)^d}{2} \] and 
	\[ \mathbb{E}\left[\left(\Phi(M')\right)^d \big\rvert Z_u \neq Z_v \right] = \frac{\left(c_u+c_v+c\right)^d+c^d}{2}.\]
	For positive integers $d$, the function $f(x) := \left(c_u+x\right)^d-x^d$ is non-decreasing for $x \geq 0$. This follows by using the binomial expansion and elementary calculus, along with the fact that $c_u \geq 0$.
	
	Because $0 \leq c \leq c_v + c$, this means \[\left(c_u+c\right)^d - c^d \leq \left(c_u+c_v+c\right)^d-\left(c_v+c\right)^d.\]
	Rearranging this inequality yields
	\[\mathbb{E}\left[\left(\Phi(M)\right)^d \big\rvert Z_u \neq Z_v \right] \leq \mathbb{E}\left[\left(\Phi(M')\right)^d \big\rvert Z_u \neq Z_v \right],\] which completes the proof.
\end{proof}

\begin{lemma}[Vertex Removal Lemma]\label{lem: vertex-removal}\ \\
	Let $M$ be a mask with the following properties, where $t$ is some positive integer.
	\begin{enumerate}
		\item The maximum $M$-degree of vertices in $V(M)$ is at most $2t$.
		\item The number of vertices in $V(M)$ is large, with $\lvert V(M) \rvert \geq 2t + 2 +  \frac{2\lvert M \rvert}{t}$.
	\end{enumerate}
	Then there exists a mask $M'$ with the following properties.
	\begin{enumerate}
		\item The maximum $M'$-degree of vertices in $V(M')$ is also at most $2t$.
		\item The number of edges in $M'$ is identical to the number of edges in $M$. That is, $\lvert M' \rvert = \lvert M \rvert$.
		\item ${\sf LDUB}(n,M) \leq  {\sf LDUB}(n,M')$.
		\item There are strictly fewer vertices in $M'$ compared to $M$. That is, $\lvert V(M') \rvert < \lvert V(M) \rvert$.
	\end{enumerate}
\end{lemma}
\begin{proof}
	Let ${\sf Low}$ be the subset of vertices in $V(M)$ with $M$-degree at most $t$. That is, ${\sf Low} := \left\{v \in V(M): {\sf deg}^M(v) \leq t \right\}$. Consider the set $V(M) \setminus {\sf Low}$. Since every vertex in this set has $M$-degree greater than $t$, there must be at least $\frac{\left(\lvert V(M) \rvert - \lvert {\sf Low} \rvert \right) \cdot t}{2}$ edges in $M$. Consequently, \[\frac{\left(2t + 2+  \frac{2\lvert M \rvert}{t} - \lvert {\sf Low} \rvert \right) \cdot t}{2} \leq \frac{\left(\lvert V(M) \rvert - \lvert {\sf Low} \rvert \right) \cdot t}{2} \leq \lvert M \rvert .\]
	Rearranging this inequality yields the useful conclusion $\lvert {\sf Low} \rvert \geq 2t+2$.
	
	Arbitrarily order the vertices in ${\sf Low}$, naming them $\{v_1,v_2,...,v_{\lvert {\sf Low} \rvert}\}$ and run the following algorithm to obtain the mask $M'$. Since $\lvert {\sf Low} \rvert \geq 2$, the algorithm is not vacuous.
	
	\textbf{initialize} $M_2 = M$
	
	\For{ $i \in \{2,...,\lvert {\sf Low} \rvert\}$}{\If{ $v_1 \in V(M_i)$}{
			$M_{i+1} = {\sf Donate}(M_i,v_1 \rightarrow v_i)$ (Algorithm~\ref{alg:donation}) }	
		{\Else{$M_{i+1} = M_i$}}}
	\textbf{output} $M' = M_{\lvert {\sf Low} \rvert+1}$
	
	For the above process to be well defined, we need every call to Algorithm~\ref{alg:donation} to be well defined. Since we always check whether $v_1 \in V(M_i)$ before invoking Algorithm~\ref{alg:donation}, and $\{v_2,...,v_{\lvert {\sf Low} \rvert}\}$ must be in $V(M_i)$ at every iteration by Fact~\ref{fact: conservation-donation}, every such call is well defined.
	
	By repeated applications of Fact~\ref{fact: conservation-donation}, it is clear that $V(M')$ must be either $V(M)$ or $V(M) \setminus v_1$.
	
	\begin{enumerate}
		\item By construction, only the vertices in ${\sf Low} \setminus v_1$ can have greater $M'$-degree than $M$-degree. However, since $v_1$, which is in ${\sf Low}$, is the only vertex that donates edges in our construction, the degree of any other vertex can increase by at most $t$. Since every vertex in ${\sf Low}$ has $M$-degree at most $t$, they can have $M'$-degree at most $2t$. Hence, the maximum $M'$-degree of any vertex in $V(M')$ is at most $2t$.
		\item $\lvert M' \rvert = \lvert M \rvert$ by repeated applications of Fact~\ref{fact: conservation-donation}.
		\item ${\sf LDUB}(n,M) \leq  {\sf LDUB}(n,M')$ by repeated applications of Lemma~\ref{lem: donation}.
		\item To show $\lvert V(M') \rvert < \lvert V(M) \rvert$, we just need to show that $v_1 \notin V(M')$.
		
		Suppose this is false, and $v_1 \in V(M')$. Then there exists a $u \in V(M') \setminus v_1$ such that $(v_1,u) \in M'$. Further, we must have taken the ``if'' branch in every iteration of our construction of $M'$, and $v_1$ must have donated edges (as specified by Algorithm~\ref{alg:donation}) to every vertex in $\{2,...,\lvert {\sf Low} \rvert\}$. Then the only way for $(v_1,u)$ to exist in $M'$ is if $(v_i,u)$ exists in $M'$ for all $v_i \in {\sf Low} \setminus u$. This means $u$ must have $M'$-degree at least $\lvert {\sf Low} \rvert -1 \geq 2t+1$. This contradicts the fact that every vertex in $V(M')$ has $M'$-degree at most $2t$. Our assumption must be false, and we must have $v_1 \notin V(M')$.
	\end{enumerate}
	This completes the proof of Lemma~\ref{lem: vertex-removal}.
\end{proof}

\begin{lemma}[Converting masks with low maximum degree to masks with few vertices]\label{lem: remove-most-vertices}\ \\
	Let $M$ be a mask where the maximum $M$-degree of any vertex in $V(M)$ is at most $2t$ for some positive integer $t$. Then there exists a mask $M'$ with the following properties.
	\begin{enumerate}
		\item ${\sf LDUB}(n,M) \leq  {\sf LDUB}(n,M')$.
		\item $M'$ has very few vertices. That is, $\vert V(M') \rvert \leq 2t + 2 + \frac{2\lvert M \rvert}{t}$.
	\end{enumerate}
	
\end{lemma}
\begin{proof}
	This follows in a straightforward manner from Lemma~\ref{lem: vertex-removal} which lets us remove vertices until the desired condition on the size of the vertex set is met.
\end{proof}

\begin{lemma}[Masks without enough vertices have small low-degree likelihood ratio]\label{lem: calculation-small-vertex}\ \\
	For any mask $M$,
	\[{\sf LDUB}(n,M) \leq 1 + \sum\limits_{d = 1}^{D} \frac{1}{d!}\cdot \left(\frac{2d \cdot \log e}{\log \left(\frac{2d \cdot n^2}{\lvert V(M) \rvert \cdot k^2}+1\right)}\right)^{2d}.\]
\end{lemma}
\begin{proof}
	For any mask $M$, it is an easy observation that adding edges to $M$ cannot decrease ${\sf LDUB}(n,M)$. Hence we have
	\[{\sf LDUB}(n,M) \leq 1 + \sum\limits_{d = 1}^{D} \frac{1}{d!}\cdot \mathbb{E}\left[\left(\sum\limits_{(i,j) \in {V(M) \choose 2}}Z_i \cdot Z_j\right)^d\right]\] where $Z_i,Z_j$ are $\{0,1\}$-valued random variables as in Definition~\ref{def: ldub}.
	
	Let $H := \sum\limits_{i \in V(M)} Z_i$. Then we have $\sum\limits_{(i,j) \in {V(M) \choose 2}}Z_i \cdot Z_j \leq \sum\limits_{i,j \in V(M)}Z_i \cdot Z_j = H^2$, and this gives
	\[{\sf LDUB}(n,M) \leq 1 + \sum\limits_{d = 1}^{D} \frac{1}{d!}\cdot \mathbb{E}\left[H^{2d}\right].\]
	Unfortunately, $H$ is the sum of dependent random variables. Since it is often easier to analyze the moments of sums of independent random variables, we use the following approach.
	
	\begin{enumerate}
		\item The $\{Z_i : i \in V(M) \}$ form a set of negatively associated (henceforth NA) random variables. This can be proved as follows.
		
		\begin{enumerate}
			\item The $\{X_i : i \in [n]\}$ (as in Definition~\ref{def: ldub}) are NA because they can be viewed as a permutation distribution~\cite[Definition 2.10 and Theorem 2.11]{joagdev}. Similarly, the $\{X'_i : i \in [n]\}$ are NA and independent of the $\{X_i : i \in [n]\}$. 
			\item $\{X_i : i \in [n]\} \cup \{X'_i : i \in [n]\}$ are jointly NA because they are the union of independent sets of NA random variables~\cite[Property P7]{joagdev}.
			\item $\{Z_i = X_i \cdot X'_i : i \in [n]\}$ are NA because they are non-decreasing functions of disjoint subsets of NA random variables~\cite[Propety P6]{joagdev}.
			\item $\{Z_i : i \in V(M)\}$ are NA because they are a subset of NA random variables~\cite[Property P4]{joagdev}.
		\end{enumerate}

		\item Negatively associated random variables can be coupled to independent random variables.
		
		Let $\{Z^*_i : i \in V(M)\}$ be a collection of independent random variables where each $Z^*_i$ has the same marginal distribution as $Z_i$. Let $H^* := \sum\limits_{i \in V(M)} Z^*_i$. For positive integers $d$, the function $f(x) = x^{2d}$ is convex when $x \geq 0$. Hence we can use \cite[Theorem 1]{shao} to conclude \[\mathbb{E}\left[H^{2d}\right] \leq \mathbb{E}\left[{(H^*)}^{2d}\right].\]
		\item Each $Z_i$ (and hence $Z^*_i$) is $1$ with probability $k^2/n^2$ and $0$ otherwise. This means $H^*$ is a Binomial random variable with $\lvert V(M) \rvert$ trials each having success probability $k^2/n^2$. We can now use known bounds on the moments of Binomial random variables (e.g.,~\cite[Corollary 1]{ahle}) and obtain \[\mathbb{E}\left[{H^*}^{2d}\right] \leq \left(\frac{2d \cdot \log e}{\log \left(\frac{2d \cdot n^2}{\lvert V(M) \rvert \cdot k^2}+1\right)}\right)^{2d}.\]
	\end{enumerate}
	Putting this all together completes the proof.
\end{proof}

\begin{lemma}[Masks with small maximum degree have small low-degree likelihood ratio]\label{lemma:nicemasks}\ \\
	Let $n = \omega(1)$ and $k \leq n$ be sequences of positive integers and $M$ be a sequence of masks on the ground set $[n]$. Suppose that the maximum $M$-degree of any vertex in $V(M)$ is small. That is, there exists a sequence of positive integers $t$ with the following properties.
	\begin{enumerate}
		\item $\max\limits_{v \in V(M)} {\sf deg}^{M}(v)  \leq 2t$.
		\item $\left(2t + 2 + \frac{2\lvert M \rvert}{t}\right) \cdot  \left(\frac{k^2}{n^2}\right) =  O(n^{-\epsilon})$ for some constant $\epsilon > 0$.
	\end{enumerate}
	Then for any sequence of degrees $D = o\left(\left(\log n\right)^2\right)$, we have the low-degree likelihood ratio upper bound
	\[{\sf LDUB}(n,M) = 1+o(1).\]
\end{lemma}
\begin{proof}
	Our first hypothesis that $\max\limits_{v \in V(M)} {\sf deg}^{M}(v)  \leq 2t$ immediately lets us combine the following:
	\begin{itemize}
		\item the simplifcation to masks without too many vertices from Lemma~\ref{lem: remove-most-vertices},
		\item the calculation of the low-degree likelihood ratio upper bound based on the number of vertices in the mask from Lemma~\ref{lem: calculation-small-vertex}.
	\end{itemize}
	Defining $v_{\sf max} := \left(2t + 2 + \frac{2\lvert M \rvert}{t}\right)$ for notational convenience, this yields
	\begin{align*}
		{\sf LDUB}(n,M) & \leq 1 + \sum\limits_{d = 1}^{D} \frac{1}{d!}\cdot \left(\frac{2d \cdot \log e}{\log \left(\frac{2d \cdot n^2}{v_{\sf max} \cdot k^2}+1\right)}\right)^{2d} \leq 1 + \sum\limits_{d = 1}^{D} \left(\frac{2d \cdot \log e}{ \left(d!\right)^{1/2d} \cdot \log \left(\frac{ n^2}{v_{\sf max} \cdot k^2}\right)}\right)^{2d} \\
		& \leq 1 + \sum\limits_{d = 1}^{D} \left(\frac{2e \cdot \sqrt{d} \cdot \log e}{ \log \left(\frac{ n^2}{v_{\sf max} \cdot k^2}\right)}\right)^{2d} \leq 1 + \sum\limits_{d = 1}^{D} \left(\frac{2e \cdot \sqrt{D} \cdot \log e}{ \log \left(\frac{ n^2}{v_{\sf max} \cdot k^2}\right)}\right)^{2d}\\
		& \leq 1 + \sum\limits_{d = 1}^{\infty} \left(\frac{2e \cdot \sqrt{D} \cdot \log e}{ \log \left(\frac{ n^2}{v_{\sf max} \cdot k^2}\right)}\right)^{2d} \leq \left(1-\left(\frac{2e \cdot \sqrt{D} \cdot \log e}{ \log \left(\frac{ n^2}{v_{\sf max} \cdot k^2}\right)}\right)^{2}\right)^{-1} \\
		& = 1+o(1).
	\end{align*}
	Above, we have used the inequality $\frac{1}{d!} \leq \left(\frac{e}{d}\right)^d$, the formula for the sum of a geometric series, the fact that $d \leq D =  o\left(\left(\log n\right)^2\right)$, and our second hypothesis $v_{\sf max}\cdot  \left(\frac{k^2}{n^2}\right) = \left(2t + 2 + \frac{2\lvert M \rvert}{t}\right) \cdot  \left(\frac{k^2}{n^2}\right) = O(n^{-\epsilon})$.
\end{proof}

\begin{lemma}[Masks without enough edges have small conditional low-degree likelihood ratio]\label{lem:cond-to-ldub}\ \\
	Let $0 < \delta < 1/2$ be a constant. Let $n = \omega(1)$ and $k = \Theta(n^{1/2+ \delta})$ be sequences of positive integers. Let $M$ be a sequence of masks on the ground set $[n]$ without too many edges.  That is,
	\[\lvert M \rvert \leq O\left(n^{\gamma} \right) \text{ for some constant }  \gamma < 3(1/2-\delta).\]
	Then there exists a sequence of subsets $S \subseteq [n]$ such that
	\begin{enumerate}
		\item $\mathop\Prob\limits_{K \sim \Clique(n,k)}\left[\text{all nonzero coordinates of }K \text{ are in } S\right] = 1-o(1)$.
		\item For any sequence of degrees $D = o(\log^2 n)$, the conditional low-degree likelihood ratio upper bound\footnote{Recall that this quantity depends on $k$ and $D$, but we do not denote this for notational simplicity.} (Definition~\ref{def: cldub}) is small: 
		\[{\sf Cond}(n,M,S) = 1 + o(1).\]
	\end{enumerate}
\end{lemma}
\begin{proof}
	Let ${\sf grow} = \log n = \omega(1)$\footnote{Any sequence that grows to infinity slower than a polynomial would work.}.
	
	Let $t := \lvert M \rvert \cdot \frac{k}{n} \cdot {\sf grow}$ and define $S \subseteq [n]$ as the subset of vertices in $M$ whose $M$-degree (Definition~\ref{def:mask}) is at most $2t$. Because there are only $\lvert M \rvert$ edges in $M$, we must have $\lvert [n] \setminus S \rvert \cdot \frac{2t}{2} \leq \lvert M \rvert$ by the pigeonhole principle. This gives $\lvert [n] \setminus S \rvert  \leq \frac{\lvert M \rvert}{t}$.
	\begin{enumerate}
		\item By a union bound, the probability that $K \sim \Clique(n,k)$ has a nonzero coordinate in $[n] \setminus S$ is at most 
		\[\lvert [n] \setminus S \rvert \cdot \frac{k}{n}  \leq \frac{\lvert M \rvert}{t} \cdot \frac{k}{n} =\frac{1}{{\sf grow}} \leq o(1).\] Thus we have \[\mathop\Prob\limits_{K \sim \Clique(n,k)}\left[\text{all nonzero coordinates of }K \text{ are in } S\right] \geq 1-o(1).\]
		\item Let $n' := \lvert S \rvert$ and fix any bijection $\phi : [n'] \rightarrow S$. Define the mask $M_S$ on ground set $[n']$ as \[M_S := \left\{(i,j) \in {[n'] \choose 2} : \left(\phi(i),\phi(j)\right) \in M \right\}.\] This is the natural restriction of the mask $M$ onto a ground set of size $n'$ corresponding to $S$.
		It is straightforward to observe that for any $n,k,D,M$ we have the following equality between a conditional low-degree likelihood ratio upper bound and a low-degree likelihood ratio upper bound:
		\[{\sf Cond}(n,M,S) = {\sf LDUB}(n',M_S).\]
		Further,
		\begin{enumerate}
			\item By construction, $n' = \Theta(n) = \omega(1)$ and $k = \Theta\left((n')^{1/2+\delta}\right)$.
			\item $\max\limits_{i \in V(M_S)} {\sf deg}^{M_S}(i)  \leq 2t$. That is, the maximum mask degree of any vertex in $V(M_S)$ is at most $2t$. This is because the $M_S$-degree of any vertex $i \in [n']$ is at most the $M$-degree of the vertex $\phi(i) \in S$, and the latter is at most $2t$ by construction of $S$.
			\item By construction, we also have $\lvert M_S \rvert \leq \lvert M \rvert \leq O\left(n^{\gamma} \right)$. Using this with the definition of ${\sf grow}$ and $t$ and the facts $0 < \delta < 1/2$, $ \gamma < 3(1/2-\delta)$ gives
			\[\left(2t + 2 + \frac{2 \lvert M_S \rvert}{t}\right) \cdot \left(\frac{k^2}{(n')^2}\right) \leq O\left(\left(n'\right)^{-\epsilon'}\right).\] for some constant $\epsilon' > 0$.
		\end{enumerate}
		This lets us invoke Lemma~\ref{lemma:nicemasks} to conclude ${\sf LDUB}(n',M_S) \leq 1+o(1)$ and complete the proof.
	\end{enumerate}
\end{proof}

\begin{proofof}{Theorem~\ref{thm:main}(a)}
	We now prove the lower bound. Following~\cite[Proposition~6.2]{fp}, to rule out weak separation by degree-$D$ polynomials, it suffices to show that the \emph{norm of the conditional low-degree likelihood ratio}
\begin{equation}\label{eq:ld-norm-def}
	\left\|\left(\frac{\mathrm{d}\mathbb{P}}{\mathrm{d}\mathbb{Q}}\right)^{\le D}\right\|_{\mathbb{Q}} := \sup_{{\sf deg}(f) \le D} \frac{\mathbb{E}_{\mathbb{P}}[f]}{\sqrt{\mathbb{E}_{\mathbb{Q}}[f^2]}}
\end{equation}
is $1+o(1)$, where $\mathbb{Q} = G(n,M)$, and $\mathbb{P}$ is $G(n,k,M)$ conditioned on some $(1-o(1))$-probability event (which may depend on latent randomness such as the clique vertices). For our purposes, we choose to condition $G(n,k,M)$ on the event that all clique vertices lie in $S$, the set defined in Lemma~\ref{lem:cond-to-ldub}. Note that Lemma~\ref{lem:cond-to-ldub} guarantees this to be a $(1-o(1))$--probability event. We have the bounds
\[ 1 \le \left\|\left(\frac{\mathrm{d}\mathbb{P}}{\mathrm{d}\mathbb{Q}}\right)^{\le D}\right\|_{\mathbb{Q}}^2 \le {\sf Cond}(n,M,S), \]
where the first inequality comes from plugging in $f = 1$ to~\eqref{eq:ld-norm-def} and the second comes from~\cite[Proposition~B.1]{spectral-planting}. Now the proof is complete, as Lemma~\ref{lem:cond-to-ldub} shows ${\sf Cond}(n,M,S) \leq 1 + o(1)$. \end{proofof}

\section{Proof of the upper bound: Theorem~\ref{thm:main}(b)}\label{sec:proof-ub}
We begin with a few preliminaries before turning to the proof of the theorem.  Our strategy is to implement the degree counting algorithms of~\cite{kuvcera1995expected,mardia2020finding} via low-degree polynomials.  To this end, we note that it suffices to furnish a mask $M$ with $\lvert M \rvert = O(n^{\gamma})$, where $\gamma > 3(1/2 - \delta)$ and a polynomial $f$ which, when evaluated on the masked observations corresponding to $M$, strongly separates (in the sense of Definition~\ref{def: separ}) the distributions $\PP = G(n, k, M)$ (see Definition~\ref{def:mask-pc}) and $\QQ = G(n,M)$ (see Definition~\ref{def:er}).  

Towards constructing this mask, we define the gap $\epsilon = \gamma - 3 (1/2 - \delta) > 0$ and the pair $R$ and $L$ as
\begin{align}\label{def:V-R}
R = \min\bigl\{\lceil (n/k)^{2} \cdot n^{2\epsilon/3} \rceil, \lceil n/2 \rceil\bigr\} \qquad \text{ and } \qquad  L = \lceil \sqrt{R} \rceil.
\end{align}
Note that $R \leq \lceil n/2 \rceil$ and $L \leq \sqrt{n} = o(n)$.  This thus ensures that the vertex sets $V_L \coloneqq \{1, 2, \ldots, L\}$ and $V_{R} \coloneqq \{n - R + 1, \ldots, n\}$ are disjoint.  We will consider the `rectangular' mask $M = V_L \times V_R$, observing that it satisfies $\lvert M \rvert = R \cdot L = O(n^{\gamma})$ by construction.

We turn now to the construction of our distinguishing polynomial $f: \{-1, 1\}^{V_L \times V_R} \rightarrow \mathbb{R}$.  In order to build intuition, let $K_1, K_2, \ldots, K_n$ denote binary indicators of whether or not vertex $i \in [n]$ belongs to the planted clique.  That is, under the null distribution $\QQ$, each of the $\{K_i\}_{i \in [n]}$ are identically zero, whereas under the planted distribution $\PP$, $\{K_i\}_{i \in [n]} \sim \mathsf{Clique}(n, k)$ (as in Definition~\ref{def:cl-ind}).  In the sequel, we show that the polynomial $(K_1, K_2, \ldots, K_L) \mapsto \sum_{i \in V_L}\, K_i$  strongly separates $\PP$ and $\QQ$.  Our distinguishing polynomial emulates this oracle polynomial by thresholding estimates related to the degree counts for each vertex in $V_L$.  To do so, we require the following lemma from~\cite{schramm2022computational}, which provides a polynomial approximation to the binary threshold function. 

\begin{lemma}[\cite{schramm2022computational}, Prop. 4.1]\label{lem:threshold-poly}
	For any integer $\ell \in \mathbb{Z}_{\geq 0}$, consider the degree-$(2\ell + 1)$ polynomial $\tau = \tau_{\ell}: \mathbb{R} \rightarrow \mathbb{R}$, $\tau(y) = (2 \ell + 1)\binom{2\ell}{\ell} \int_{0}^{y} t^\ell (1 - t)^{\ell} \mathrm{d}t$.  For any $b \in \{0, 1\}$ and any $0 \leq \Delta \leq \frac{1}{2}$, the following holds
	\[
	\bigl \lvert \tau(y) - b \bigr \rvert \leq \Bigl(\ell + \frac{1}{2}\Bigr) (6\Delta)^{\ell}, \qquad \text{ for any } y \in \mathbb{R} \text{ such that } \qquad \lvert y - b \rvert \leq \Delta.
	\]
\end{lemma}
Let $Y = \{Y_{ij}\}_{i \in [V_L], j \in [V_R]} \in \{-1, 1\}^{V_L \times V_R}$ denote our observations.  Using $\tau = \tau_{\ell}$ as defined in Lemma~\ref{lem:threshold-poly}, with $ \ell = \lceil 3/\epsilon + 1/\delta\rceil$, we define our separating polynomial $f: \{-1, 1\}^{L \times R} \rightarrow \mathbb{R}$ as
\begin{align}
	\label{def:separating-polynomial}
	f(Y) = \sum_{i \in V_L} \tau\Bigl(\frac{n}{k} \cdot\frac{1}{R}\sum_{j \in V_R} Y_{ij} \Bigr).
\end{align}
For convenience, we will use the shorthand
\[
g_i(Y) = \frac{n}{k} \cdot \frac{1}{R}\sum_{j \in V_R} Y_{ij} .
\]
The key property of our polynomial $f$ is that $\tau(g_i(Y))$ emulates the clique indicators $K_i$ up to a small error, as summarized by the following lemma.
\begin{lemma}
	\label{lem:second-moment-bound-tau}
	 Under the conditions of Theorem~\ref{thm:main}, let $L$ and $R$ be as in~\eqref{def:V-R} and consider the mask $M = V_L \times V_R$.  Let $Y \in \{-1, 1\}^{M}$ denote the observations. 
  Suppose that $ \ell = \lceil 3/\epsilon + 1/\delta\rceil$ and let $\tau = \tau_{\ell}$ be as defined in Lemma~\ref{lem:threshold-poly}.  Then, if either $Y \sim \PP = G(n, k, M)$ or $Y \sim \QQ = G(n, M)$, both
	\begin{align} \label{ineq:second-moment-bound-tau}
		\EE_{\PP}\bigl[\{\tau(g_1(Y)) - K_1\}^2\bigr] = o\Bigl(\Bigl(\frac{k}{n}\Bigr)^2\Bigr)\qquad \text{ and } \qquad 
		\EE_{\QQ}\bigl[\{\tau(g_1(Y))\}^2\bigr] = o\Bigl(\Bigl(\frac{k}{n}\Bigr)^{2}\Bigr).
	\end{align}
\end{lemma}

We defer the proof of this lemma to the end of the section.  Equipped with this lemma, we turn to the proof of Theorem~\ref{thm:main}(b).  

\begin{proofof}{Theorem~\ref{thm:main}(b)}
We turn to lower bounding the expectation gap and upper bounding the variance induced by the polynomial $f$ in~\eqref{def:separating-polynomial}.

\noindent \underline{Lower bounding the expectation gap:} 
    Expanding yields 
\begin{align*}
	\bigl \lvert \EE_{\PP}[f(Y)] - \EE_{\QQ}[f(Y)] \bigr \rvert &= L \bigl \lvert \EE_{\PP}[K_1] +  \EE_{\PP}[\tau(g_1(Y)) - K_1] -  \EE_{\QQ}[\tau(g_1(Y)) - K_1] \bigr \rvert \\
	&= L \Bigl \lvert \frac{k}{n} +  \EE_{\PP}[\tau(g_1(Y)) - K_1] -  \EE_{\QQ}[\tau(g_1(Y))] \Bigr \rvert,
\end{align*}
where we have used the fact that $K_1 = 0$ under the null distribution $\QQ$.  
Then, applying the triangle inequality in conjunction with Jensen's inequality yields
\begin{align*}
	\bigl \lvert \EE_{\PP}[f(Y)] - \EE_{\QQ}[f(Y)] \bigr \rvert  &\geq \frac{Lk}{n} - L \bigl \lvert \EE_{\PP}[g_1(Y) - K_1 ] - \EE_{\QQ} [g_1(Y)] \bigr \rvert \\
	&\geq \frac{Lk}{n} - L\sqrt{ \EE_{\PP}\bigl[ \{g_1(Y) - K_1 \}^2 \bigr]} - L\sqrt{\EE_{\QQ}\bigl[ \{g_1(Y)\}^2 \bigr]}.
\end{align*}
We conclude by applying Lemma~\ref{lem:second-moment-bound-tau} to obtain the bound
\[
\bigl \lvert \EE_{\PP}[f(Y)] - \EE_{\QQ}[f(Y)] \bigr \rvert = \frac{Lk}{n} - o\Bigl(\frac{Lk}{n}\Bigr) = \Omega\bigl( \min\{n^{\delta}, n^{\epsilon/3}\}\bigr).
\] 

\medskip
\noindent \underline{Upper bounding the variance:}

\noindent \textit{Planted distribution:}
\[
\Var_{\PP}(f) = \Var_{\PP} \biggl(\sum_{i \in V_L} [\tau(g_i(Y)) - K_i] + K_i \biggr) \leq 2\Var_{\PP} \biggl(\sum_{i \in V_L} [\tau(g_i(Y)) - K_i]\biggr) + 2 \Var_{\PP} \biggl(\sum_{i \in V_L} K_i\biggr).
\]
To bound the first term, we apply Lemma~\ref{lem:second-moment-bound-tau} to obtain
\[
2\Var_{\PP} \biggl(\sum_{i \in V_L} [\tau(g_i(Y)) - K_i]\biggr) \leq 2L^2 \EE_{\PP}\bigl[\{\tau(g_1(Y)) - K_1\}^2\bigr] = o\bigl(\bigl[ \min\{n^{\delta}, n^{\epsilon/3}\}\bigr]^2\bigr). \] 

Moreover, we compute
\begin{align*}
2 \Var_{\PP} \biggl(\sum_{i \in V_L} K_i\biggr) = 2L \Var_{\PP}(K_1) + 2 L (L-1) \mathsf{Cov}_{\PP}(K_1, K_2) &\leq 2L \Bigl(\frac{k}{n}\Bigr) + 2L (L-1) \Bigl(\frac{k \cdot (k-1)}{n \cdot (n -1)} - \frac{k^2}{n^2}\Bigr) \\
&= O\bigl( \min\{n^{\delta}, n^{\epsilon/3}\}\bigr).
\end{align*}
\noindent \textit{Null distribution:}
Proceeding similarly yields 
\[
\Var_{\QQ}(f) \leq 2L^2 \EE_{\QQ}[\tau(g_1(Y))^2] =  o\bigl(\bigl[ \min\{n^{\delta}, n^{\epsilon/3}\}\bigr]^2\bigr).
\]

Putting the pieces together then yields
\[
\Bigl(\max\bigl\{\Var_{\PP}(f), \Var_{\QQ}(f)\bigr\}\Bigr)^{1/2} = o\bigl(\min\{n^{\delta}, n^{\epsilon/3}\}\bigr) = o\bigr(\bigl \lvert \EE_{\PP}[f(Y)] - \EE_{\QQ}[f(Y)] \bigr \rvert\bigr),
\]
which confirms that $f$ strongly separates $\PP$ and $\QQ$.  
\end{proofof}

\begin{proofof}{Lemma~\ref{lem:second-moment-bound-tau}}
		Note that, conditioned on the event that vertex $i$ is not contained in the clique, $\EE_{\PP}[g_i(Y) \; \vert \; K_i = 0] = 0$, whereas conditioned on the event that vertex $i$ is contained in the clique, $\EE_{\PP}[g_i(Y) \; \vert \; K_i = 1] = 1$.  Then, conditioned on $K_i = 0$, by Bernstein's inequality (e.g.,~\cite[Theorem  2.8.4]{vershynin2018high}),
		\begin{align*}
			\PP\biggl(\biggl \lvert \sum_{j \in V_R}\, Y_{ij}  \biggr \rvert \geq t \; \Big \vert \; K_i = 0\biggr) \leq 2 \exp\Bigl(-\frac{t^2/2}{R + t/3}\Bigr).
		\end{align*}
		We next express $Y_{ij}$ as $Y_{ij} = (1 - K_i K_j) A_{ij} + K_i K_j$, where $A_{ij} \overset{\mathsf{i.i.d.}}{\sim} \mathsf{Rademacher}(1/2)$ for $1 \leq i < j \leq n$.  Then, conditionally on the event $\{K_i = 1\}$, the collection $\{Y_{ij}\}_{j \in V_R}$ are monotone functions of the negatively associated random variables $\{K_j\}_{j \in V_R}$ and the independent random variables $\{A_{ij}\}_{j \in V_R}$.  Consequently, by~\cite[Property P6]{joagdev} we find that $\{Y_{ij}\}_{j \in V_R}$ forms a negatively associated collection, conditionally on the event $\{K_i = 1\}$.
		Next, let $\{Y^{\ast}_{ij}\}_{j \in V_R}$ denote a collection of independent random variables with the same marginal distributions as $Y_{ij}$.  Further let $S = \sum_{j \in V_R} Y_{ij} - \EE[Y_{ij} \; \vert \; K_i = 1]$ and define $S^{\ast}$ similarly.  Then, applying~\cite[Theorem 1]{shao}, we obtain the MGF bound $\EE[\exp\{\lambda S\} \; \vert \; K_i = 1] \leq \EE[\exp\{\lambda S^\ast\}]$, for all $\lambda \in \mathbb{R}$ such that the RHS exists.
		The discussion above thus shows that Bernstein's inequality for bounded random variables~\cite[Theorem 2.8.4]{vershynin2018high} continues to hold for the collection $\{Y_{ij}\}_{j \in V_R}$, whence we obtain the inequality 
		\begin{align*}
			\PP\biggl(\biggl \lvert \sum_{j \in V_R}\, Y_{ij} - R \cdot \frac{k}{n} \biggr \rvert \geq t \; \Big \vert \; K_i = 1\biggr) \leq 2 \exp\biggl(-\frac{t^2/2}{R + 2t/3}\biggr).
		\end{align*}
		Combining the previous two displays yields the inequality
		\begin{align}\label{ineq:vertex-tail-bound}
			\PP\biggl(\bigl\lvert g_i(Y) - K_i \bigr \rvert \geq t \cdot \frac{n}{kR} \biggr) \leq 2 \exp\biggl(-\frac{t^2/2}{R + 2t/3}\biggr).
		\end{align}
		Equipped with this concentration inequality, we turn to bounding the second moment $\EE_{\PP}[\{\tau(g_1(Y)) - K_1\}^2]$, which we decompose as 
		\begin{align}
			\EE_{\PP}\bigl[\{\tau(g_1(Y)) - K_1\}^2\bigr] &= \EE_{\PP}\bigl[\{\tau(g_1(Y)) - K_1\}^2\mathbbm{1}\{\lvert g_1(Y) - K_1 \rvert \leq 1/2\}\bigr] \nonumber\\
			&\qquad \qquad + \EE_{\PP}\bigl[\{\tau(g_1(Y)) - K_1\}^2\mathbbm{1}\{\lvert g_1(Y) - K_1 \rvert > 1/2\}\bigr].
		\end{align}
		We claim the following two upper bounds, deferring their proofs to the end
		\begin{subequations}
			\begin{align}
				\EE_{\PP}\bigl[\{\tau(g_1(Y)) - K_1\}^2\mathbbm{1}\{\lvert g_1(Y) - K_1 \rvert \leq 1/2\}\bigr] &\leq C_{\ell} \cdot \max\Bigl\{n^{-\epsilon \ell/3}, n^{-\delta \ell}\Bigr\}\label{ineq:truncated-lower}\\
				\EE_{\PP}\bigl[\{\tau(g_1(Y)) - K_1\}^2\mathbbm{1}\{\lvert g_1(Y) - K_1 \rvert > 1/2\}\bigr] &\leq C_{\ell} \cdot n^{4\ell + 2}\max\Bigl\{e^{-c n^{2\epsilon/3}}, e^{-cn^{2\delta}} \Bigr\}, \label{ineq:truncated-upper}
			\end{align}
		\end{subequations}
		where $C_{\ell}$ denotes a constant which depends only on $\ell$ and may change line by line.  Then, taking $n$ large enough to ensure that the RHS of inquality~\eqref{ineq:truncated-upper} is upper bounded by $C_{\ell} \cdot \max\{n^{-\epsilon \ell/3}, n^{-\delta \ell}\}$ and repeating similar steps under $\QQ$, we obtain the pair of bounds
		\begin{align*}
			\EE_{\PP}\bigl[\{\tau(g_1(Y)) - K_1\}^2\bigr] &\leq C_{\ell} \cdot \max\Bigl\{n^{-\epsilon \ell/3}, n^{-\delta \ell}\Bigr\} \qquad \text{ and } \nonumber \\
			\EE_{\QQ}\bigl[\{\tau(g_1(Y))\}^2\bigr] &\leq C_{\ell} \cdot \max\Bigl\{n^{-\epsilon \ell/3}, n^{-\delta \ell}\Bigr\},
		\end{align*}
	as desired.  The desired result follows by noting that $ \ell \geq \max\{3/\epsilon, 1/\delta\}$ and $1/n = o((k/n)^2)$.  It remains to establish the pair of inequalities~\eqref{ineq:truncated-lower} and~\eqref{ineq:truncated-upper}.
 
	\paragraph{Proof of the inequality~\eqref{ineq:truncated-lower}.}
	Applying Lemma~\ref{lem:threshold-poly} yields 
	\begin{align*}
		\EE_{\PP}\bigl[\{\tau(g_1(Y)) - K_1\}^2\mathbbm{1}\{\lvert g_1(Y) - K_1 \rvert \leq 1/2\}\bigr] &\leq 6^{\ell}(\ell + 1/2)\EE_{\PP}\bigl[\lvert g_1(Y) - K_1 \rvert^{\ell}\mathbbm{1}\{\lvert g_1(Y) - K_1 \rvert \leq 1/2\}\bigr]\\
		&\leq  6^{\ell}(\ell + 1/2)\EE_{\PP}\bigl[\lvert g_1(Y) - K_1 \rvert^{\ell}\bigr].
	\end{align*}
	We then integrate and apply the tail bound~\eqref{ineq:vertex-tail-bound} to obtain the inequality
	\begin{align*}
		\EE_{\PP}\bigl[\lvert g_1(Y) - K_1 \rvert^{\ell}\bigr] &= \int_{0}^{\infty}\; \ell t^{\ell - 1} \PP\{\lvert g_1(Y) - K_1 \rvert \geq t\} \mathrm{d}t \leq 4  \int_{0}^{\infty}\; \ell t^{\ell - 1} \exp\Bigl\{-\frac{ct^2 k^2 R}{n^2 + tkn}\Bigr\} \mathrm{d}t \\
		&\leq \int_{0}^{C n/k}\, \ell t^{\ell - 1} \exp\Bigl\{-\frac{ct^2 k^2 R}{n^2}\Bigr\} \mathrm{d}t +  \int_{0}^{\infty}\, \ell t^{\ell - 1} \exp\Bigl\{-\frac{ct k R}{n}\Bigr\} \mathrm{d}t \\
		&= C_{\ell} \Bigl(\frac{n}{k\sqrt{R}}\Bigr)^{\ell} \cdot \int_{0}^{C\sqrt{R}}\, \ell u^{\ell - 1} e^{-u^2} \mathrm{d}u + C_{\ell} \Bigl(\frac{n}{kR}\Bigr)^{\ell} \cdot \int_{0}^{\infty}\, \ell u^{\ell - 1} e^{-u} \mathrm{d}u \\
		&\leq C_{\ell} \Bigl(\frac{n}{k\sqrt{R}}\Bigr)^{\ell} \cdot \bigl\{\Gamma(\ell/2) + \Gamma(\ell)\bigr\} \leq C_{\ell} \Bigl(\frac{n}{k\sqrt{R}}\Bigr)^{\ell} \leq C_{\ell} \cdot \max\Bigl\{n^{-\epsilon \ell/3}, n^{-\delta \ell}\Bigr\},
	\end{align*}
	where $\Gamma(\cdot)$ denotes the Gamma function, the penultimate inequality follows from the numeric inequality $\Gamma(\ell) \leq 3\ell^\ell$ for all $\ell \geq 1/2$ and the final inequality follows from substituting $R = \min\{\lceil (n/k)^{2} \cdot n^{2\epsilon/3}\rceil, n/2\}$.

	\paragraph{Proof of the inequality~\eqref{ineq:truncated-upper}.}
	Applying the Cauchy--Schwarz inequality yields
	\begin{align*}
		\EE_{\PP}\bigl[\{\tau(g_1(Y)) - K_1\}^2\mathbbm{1}\{\lvert g_1(Y) - K_1 \rvert > 1/2\}\bigr] &\leq \EE_{\PP}\bigl[\{\tau(g_1(Y)) - K_1\}^4\bigr]^{1/2} \PP\{\lvert g_1(Y) - K_1 \rvert > 1/2\}^{1/2} \\
		&\overset{\1}{\leq} \sqrt{2}\EE_{\PP}\bigl[\{\tau(g_1(Y)) - K_1\}^4\bigr]^{1/2} \max\Bigl\{e^{-c n^{2\epsilon/3}}, e^{-cn^{2\delta}} \Bigr\}\\
		&\overset{\2}{\leq} 4 \bigl(\EE_{\PP}[\tau(g_1(Y))^4] + 1\bigr)^{1/2} \max\Bigl\{e^{-c n^{2\epsilon/3}}, e^{-cn^{2\delta}} \Bigr\}\\
		&\leq C_{\ell} n^{4\ell + 2} \max\Bigl\{e^{-c n^{2\epsilon/3}}, e^{-cn^{2\delta}} \Bigr\},
	\end{align*}
	where step $\1$ follows from the inequality~\eqref{ineq:vertex-tail-bound} and step $\2$ follows from the numeric inequality $(a - b)^4 \leq 8a^4 + 8b^4$.
\end{proofof}

\section*{Acknowledgments}
We are thankful to the Simons Institute for the Theory of Computing for their hospitality during Fall 2021, where this work was initiated. K.A.V.\ was supported by European Research Council Advanced Grant 101019498.  A.S.W.\ was partially supported by an Alfred P.\ Sloan Research Fellowship and NSF CAREER Award CCF-2338091.

\bibliographystyle{alpha}
\bibliography{ref}

\end{document}